%% file: Lazy_Reimplication_in_Chronological_Backtracking.tex
\documentclass[a4paper,USenglish,cleveref, autoref, thm-restate]{lipics-v2021}
\usepackage[utf8]{inputenc}
\usepackage{mathtools}

\usepackage{amsmath,amssymb,cancel,tikz,bussproofs}
\usepackage{multicol,hyperref,multirow}
\usepackage{todonotes}

\usepackage{cite}

\usepackage{mathpartir}

\usepackage{booktabs}
\usepackage{arydshln}
\usepackage{siunitx}
\usepackage{mdframed,xspace}
\usepackage{algorithm}
\usepackage[noend]{algpseudocode}
\newcommand{\Continue}{\textbf{continue} }

\usepackage{pifont}
\newcommand{\cmark}{\ding{51}}
\newcommand{\xmark}{\ding{55}}

\usepackage[misc,geometry]{ifsym}
\usepackage{microtype}
\microtypesetup{final,tracking=true,kerning=true,spacing=true,stretch=10,shrink=10,letterspace=25}
\microtypecontext{spacing=nonfrench}

\hypersetup{
   pdfkeywords={},
   pdfborder={0 0 0},
   draft=false,
   bookmarksnumbered,
   bookmarksdepth=2,
   bookmarksopenlevel=2,
   bookmarksopen}

\DeclareSIUnit\loc{loc}

\input{sat-macros.tex}

\newtheorem{invariant}[theorem]{Invariant}

\def\orcidID#1{\href{http://orcid.org/#1}{\raisebox{-1.25pt}{\includegraphics{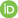}}}}

\renewcommand{\paragraph}[1]{\par\smallskip\noindent\textbf{#1}}

\newcommand{\topclause}{\blacksquare}
\renewcommand{\square}{\topclause}

\title{Lazy Reimplication in Chronological~Backtracking}
\titlerunning{Lazy Reimplication in Chronological Backtracking}

\author{%
    Robin Coutelier}
    {TU Wien, Vienna, Austria}
    {robin.coutelier@tuwien.ac.at}
    {0009-0002-4735-5215}
    {}

\author{%
    Mathias Fleury}
    {University Freiburg, Freiburg, Germany}
    {fleury@cs.uni-freiburg.de}
    {0000-0002-1705-3083}
    {}

\author{%
    Laura Kovács}
    {TU Wien, Vienna, Austria}
    {laura.kovacs@tuwien.ac.at}
    {0000-0002-8299-2714}
    {}

\authorrunning{Coutelier, Fleury, and Kovács}
\date{}

\Copyright{Robin Coutelier, Mathias Fleury, and Laura Kovács}
\EventEditors{Supratik Chakraborty and Jie-Hong Roland Jiang}
\EventNoEds{2}
\EventLongTitle{27th International Conference on Theory and Applications of Satisfiability Testing (SAT 2024)}
\EventShortTitle{SAT 2024}
\EventAcronym{SAT}
\EventYear{2024}
\EventDate{August 21--24, 2024}
\EventLocation{Pune, India}
\EventLogo{}
\SeriesVolume{305}
\ArticleNo{9}

\supplementdetails[subcategory={Source Code}]{\cadical}{https://github.com/arminbiere/cadical/tree/strong-backtrack}
\supplementdetails[subcategory={Source Code GitLab}]{\napsat}{https://gitlab.uliege.be/smt-modules/sat-library}
\supplementdetails[subcategory={Source Code GitHub}]{\napsat}{https://github.com/RobCoutel/NapSAT}
\supplementdetails[subcategory={Source Code}]{\glucose}{https://github.com/m-fleury/glucose}

\acknowledgements{
The first author's Master thesis \cite{coutelier2023chronological} conducted at the University of Liège under the supervision of Pascal Fontaine laid the foundations of this work. We thank him for his insights.}

\funding{
The authors acknowledge support from the ERC Consolidator Grant ARTIST 101002685; the TU Wien Doctoral College TrustACPS; the FWF SpyCoDe SFB projects \mbox{F8504}; the WWTF Grant ForSmart 10.47379/ICT22007; the
and the Amazon Research Award 2023 QuAT; the bwHPC of the state of Baden-Württemberg; the German Research Foundation (DFG) through grant INST 35/1597-1 FUGG; and a gift from Intel Corporation.}

\nolinenumbers
\begin{document}

\maketitle

\begin{abstract}
Chronological backtracking is an interesting SAT solving technique within CDCL reasoning, as it backtracks less aggressively upon conflicts. However, chronological backtracking is more difficult to maintain due to its weaker SAT solving invariants. This paper introduces a lazy reimplication procedure for missed lower implications in chronological backtracking. Our method saves propagations by reimplying literals on demand, rather than eagerly. Due to its modularity, our work can be replicated in other solvers, as shown by our results in the solvers \cadical and \glucose.
\end{abstract}

\section{Introduction}
\label{sec:intro}
In the past few years, chronological backtracking in CDCL-based SAT solving attracted renewed interest
as it implements less aggressive procedures when backtracking upon conflicts, particularly for
undoing literal assignments stored in the assignment stack. Chronological backtracking has been proven sound and complete, while also empirically improving performance on SAT competition problems~\cite{DBLP:conf/sat/NadelR18,DBLP:conf/sat/MohleB19,DBLP:conf/sat/Nadel22}.

Without chronological backtracking in SAT solving, the truth value of each literal is set as early as possible in the solving process. With chronological backtracking,
there are, however, missed lower implications (MLI), i.e., clauses that could have set a literal at a lower SAT decision level.
As a remedy to MLI, IntelSAT~\cite{DBLP:conf/sat/Nadel22} and \cadical-1.9.4~\cite{cadical-reimply} fix
the level of the assignments. Modifying levels impacts solving performance and significantly clutters the code; for example, reimplication techniques for detecting MLI have been removed in \cadical-1.9.5~\cite{cadical}
due to the increased code complexity.

In this paper, we introduce a \emph{lazy reimplication procedure for resolving missed lower implications in chronological backtracking, while also ensuring efficiency in SAT solving.} Doing so, in \Cref{fig:invariant-list}, we state the invariant properties to be maintained during CDCL and highlight differences between relevant backtracking approaches in SAT solving. In particular, we consider and adjust variants of \emph{non-chronological backtracking} (NCB)~\cite{DBLP:journals/tc/Marques-SilvaS99} and
strong chronological backtracking (SCB)~\cite{DBLP:conf/sat/Nadel22}.
A formal presentation of these invariants and backtracking variants is given in \Cref{sec:invariant-framework}. Using the invariants of \Cref{fig:invariant-list}, in \Cref{sec:practical-approach}
we introduce a lazy reimplication procedure to handle missed lower implications, with a particular focus on handling unit implications after backtracking. We also
adjust and enhance the
first unique implication point (UIP) algorithm~\cite{DBLP:conf/dac/MoskewiczMZZM01} with the knowledge of missed lower implications. Our approach is sound (\Cref{sec:soundness}). We implemented our work in the new solver \napsat\cite{napsat} and present our empirical findings in \Cref{sec:empirical-results}. To demonstrate the flexibility of our lazy reimplication techniques, we also implemented the algorithms of \Cref{sec:practical-approach} in
\cadical~\cite{BiereFleuryPolitt-SAT-Competition-2023-solvers} and \glucose~\cite{glucose}, and provide empirical comparisons using these solvers.

\begin{figure}[t]
    \centering
    \renewcommand{\arraystretch}{2.5}
    \begin{mdframed}
    \begin{tabular}{wl{5.2cm} p{7cm}}
        \multicolumn{2}{l}{Invariant properties for CDCL algorithms (\Cref{sec:invariant-framework})}\\
        \hdashline
        \Cref{inv:weak-watched-literals}--\emph{Weak watched literals}:
        & \parbox{0.55\textwidth}{No conflict is missed.} \\
        \midrule
        \multicolumn{2}{l}{Invariants on implications, native for NCB and WCB (\Cref{sec:convenient-cdcl-invariants})}\\
        \hdashline
        \Cref{inv:implied-literals}--\emph{Implied literals}:
        & Literals are decisions or implied by a clause $C$ that is made unit by the partial assignment.\\
        \Cref{inv:topological-order}--\emph{Topological order}:
        & The partial SAT assignment follows a topological order of the implication graph. \tabularnewline[0.5em]
        \midrule
        \multicolumn{2}{l}{Strong invariant, non-trivial for CDCL with CB and native in NCB (\Cref{sec:convenient-cdcl-invariants})}\\
        \hdashline
        \rule{0pt}{6ex}
        \Cref{inv:strong-watched-literals}--\emph{Strong watched literals}:
        & No implication nor conflict can be missed. \\
    \end{tabular}
    \end{mdframed}
    \caption{Invariant properties for CDCL-based SAT solving and maintained by the different chronological backtracking (CB) strategies, particularly by
    non-chronological backtracking (NCB) and weak chronological backtracking (WCB).
    \label{fig:invariant-list}}
\end{figure}

\paragraph{Related work.}
Within CDCL, the truth values of literals are assigned by guessing (deciding) and propagating them in a trail until a conflict is found. Upon conflict analysis, the trail is adapted by backtracking, i.e. revoking some assignments and swapping the truth value of one variable, called the unique implication point (UIP). The standard approach~\cite{DBLP:journals/tc/Marques-SilvaS99} is to fix the conflict as early as possible with \emph{non-chronological backtracking} (NCB) and all assignments between the current point and the point where the UIP is set are deleted.

A different backtracking approach comes with \emph{chronological backtracking} (CB)~\cite{DBLP:conf/sat/NadelR18,DBLP:conf/sat/MohleB19}. Here, a less
aggressive backtracking scheme is used and some propagations and decisions are kept. Chronological backtracking may backjump at any level between the UIP and the
UIP falsification point minus one. As a result, chronological backtracking resets a smaller part of the trail, but it may miss propagations
that could have been done earlier if the learned clause was known beforehand. In this paper, we refer by \emph{weak chronological backtracking} (WCB) to the CDCL algorithms that use chronological backtracking mechanism and which do not detect every propagation as early as possible (see \Cref{sec:invariants-in-cb})

For recovering such missed propagations, we define \emph{strong chronological backtracking} (SCB). In particular, Nadel~\cite{DBLP:conf/sat/Nadel22} introduced a reimplication procedure that eagerly re-assigns literals detected as missed lower implications to their lowest possible level.
We refer to this SCB technique as \emph{eager strong chronological backtracking} (ESCB).
Our work introduces a new SCB method, \emph{lazy strong chronological backtracking} (LSCB). Unlike ESCB, within LSCB we reimply missed implications on demand. As such, our work is stronger than WCB, as WCB does not perform reimplications at all. In addition, our technique is shown to be easier and more flexible to implement than ESCB or WCB (\Cref{sec:empirical-results}).

\paragraph{Our contributions.} This paper brings the following contributions to chronological backtracking in CDCL-based proof search.
\begin{enumerate}
    \item We formalize invariant properties that need to be maintained during SAT solving with chronological backtracking (\Cref{sec:invariant-framework}). Our invariants incorporate and reason over different backtracking strategies.
    \item We introduce
    \emph{lazy strong chronological backtracking} (LSCB) for on-demand reimplication of missed lower implications (Sec.~\ref{sec:practical-approach}) and prove soundness of our approach (\Cref{sec:soundness}).
    \item We implement our work in the new \napsat\cite{napsat,modularit} solver (\Cref{sec:empirical-results}).
    To showcase the flexibility and efficiency of our approach, we integrate LSCB into \cadical~\cite{BiereFleuryPolitt-SAT-Competition-2023-solvers} and \glucose~\cite{glucose}, and provide experimental comparisons using these solvers.
\end{enumerate}

\section{Preliminaries}
\label{sec:preliminaries}

We assume familiarity with propositional logic and CDCL~\cite{DBLP:series/faia/336}, and use the standard logical connectives $\neg$, $\land$, and $\lor$.
A finite set of elements (e.g.\, literals) is called \emph{conjunctive} (respectively, \emph{disjunctive}) to indicate that the set is the conjunction (respectively, disjunction) of its elements.
An \emph{ordered set} is a set $\mathcal{S}$ which defines a bijective function $p_{\mathcal{S}}$ from elements of $\mathcal{S}$ to naturals, such that $p_{\mathcal{S}}(e)$ is the position of the element $e$ in the ordered set $\mathcal{S}$. We consider the first element of $\mathcal{S}$ to have the position $0$. Ordered sets are stable under the removal of elements; that is, for the ordered sets $\mathcal{S}, \mathcal{T}, \mathcal{U}$ with $\mathcal{S} = \mathcal{T} \setminus \mathcal{U}$, we have $\forall e, e' \in \mathcal{S}.\ p_\mathcal{S}(e) < p_\mathcal{S}(e') \Leftrightarrow p_\mathcal{T}(e) < p_\mathcal{T}(e')$. We denote by $\cdot$ set concatenation; for simplicity, we use $\cdot$ to also denote appending a sequence with an element.
We write $\mathcal{S}[a:b]$ to select the ordered elements $e$ in $\mathcal{S}$ with positions $a \leq p_\mathcal{S}(e) \leq b$.

We denote by $\mathcal{V}$ a countable set of Boolean variables $v$.
We consider propositional formulas $F$ in conjunctive normal form (CNF), represented by a conjunctive set of clauses $\{C_1, C_2, \ldots, C_n\}$ over $\mathcal{V}$.
Clauses are disjunctive sets of literals $C = \{c_1, c_2, \ldots, c_m\}$, where a literal $c_i$ is either a Boolean variable $v$ or
a negation $\neg v$ of a variable $v$.

To efficiently identify unit propagations, SAT solvers track two literals per clause in the two-watched literal scheme~\cite{DBLP:conf/dac/MoskewiczMZZM01}. We denote the watched literals of a clause $C$ by
$c_1$ and $c_2$,
and write $\wl(c_1)$ and $\wl(c_2)$ for the watched lists of $c_1$ and $c_2$. We have $C \in \wl(c_1)\cap\wl(c_2)$.

During SAT solving, solvers keep track of a \emph{partial assignment}, also called \emph{trail} and denoted as the conjunctive ordered set $\p = \trail \cdot \q$, which is split into two parts:
(i) $\trail$ is the set of literals that were already propagated and do not need to be inspected anymore (by checking the watcher lists);
(ii) $\q$ is the \emph{propagation queue} containing literals that were implied and waiting to be propagated.
The partial assignment $\p$ contains the set $\p^d \subseteq \p$ of decision literals. Decisions literals in $\p^d$  are arbitrarily chosen literals when unit propagation cannot be further used and the truth value of a (decision) literal needs to be picked and assigned.
We call \emph{unit}, a clause $C$ containing exactly one unassigned literal $\ell$ and whose other literals are falsified, i.e., $\exists \ell \in C.\ C \setminus \{\ell\}, \p \models \bot \land \{\ell, \neg \ell\} \cap \p = \emptyset$.
For conflict analysis, the propagation reasons of literals are analyzed. Therefore, SAT solvers use a $\reason$ function that maps literals to clauses such that $\reason(\ell)$ captures the reason for propagating $\ell$.
The reason for propagating $\ell$ is the clause $C$ that implied $\ell$ under assumption $\p$, that is, $[\ell \in \p] \land [\ell \in \reason(\ell)] \land [\reason(\ell) \setminus \{\ell\}, \p \models \bot]$.
Following~\cite{DBLP:conf/sat/MohleB19}, we use $\level$ to represent the (decision) \emph{level of $\ell$}, i.e., the level when a truth assignment to $\ell$ was made. Formally, if $\ell$ is a decision literal, then the level $\level(\ell) = \level(\neg \ell)$ of $\ell$ is
the number of decisions preceding and including $\ell$, that is $\level(\ell) = |\p[0:p_\p(\ell)] \cap \p^d|$. Further, for literals $\ell$ implied by $\reason(\ell)$, we have
$\level(\ell) = \max_{\ell' \in \reason(\ell)\setminus\{\ell\}} \level(\ell')$. Finally, $\level(\ell) = \infty$ for unassigned literals $\ell$. The definition of $\level$ is extended to clauses and trails, with $\level(C)=\max_{\ell \in C}\level(\ell)$; similarly for $\level(\p)$. The level of the empty set is $\level(\emptyset) = 0$. We write $\delta[\ell \gets d]$ to denote that the level of $\ell$ is updated to $d$.
We reserve the special symbol $\topclause$ to denote \emph{undefined clauses} during SAT solving, with $\level(\square) = \infty$.

In standard CDCL with
non-chronological backtracking (NCB)~\cite{DBLP:journals/tc/Marques-SilvaS99}, level \(\level\) stores the number of decisions that appear before in the trail, and is always the lowest level possible.
In CDCL with chronological backtracking, the history of propagations and conflicts may, however, lead to \emph{missed lower implications} (MLI), where a MLI captures the fact that a clause $C$ is satisfied by a unique literal $\ell$ at a level \emph{strictly higher} than $\level(C \setminus \{\ell\})$. Therefore, in an MLI, the literal \(\ell\)
could have been propagated at a lower level in the trail.

\begin{example}[Missed Lower Implications -- MLI]\label{ex:MLI}
\Cref{fig:mli-ex} shows a clause set $\{C_1,\ldots,C_7\}$ and a trail $\p = \trail \cdot \q$ during CDCL solving with chronological backtracking. The trail diagram displays, from left to right, the order in which literals are decided and propagated, as well as the location of the propagation head (symbolized by the dashed line). The propagation level is symbolized by the height of the step. As a visual aid, literals are colored \green{green}, \red{red} or black, symbolizing respectively satisfied, falsified, and unassigned literals. The watched literals are the first two in the clause.

In the example, the clause set $F_0 = \{C_1, \dots, C_6\}$ was given as input. The decisions $v_1, v_2$ and $v_3$ were made, then the solver found a conflict in $C_2$ after implying $v_4$ with reason $\reason(v_4) = C_1$. The clause $C_7 = \red{\neg v_{3}}\onLevel{3} \lor \red{\neg v_{1}}\onLevel{1}$ is learned and the solver backtracks to level $\level(C_7) - 1 = 2$, continuing its propagations until it reaches the assignment shown on \Cref{fig:mli-ex}.
\Cref{fig:mli-ex} shows that $C_4$ is a MLI. Indeed, $v_2$ is satisfied at level~2, while all other literals are falsified at level~1. After backtracking to level 1, the implication of $v_2$ by $C_4$ is missed since (i) $\neg v_3$ was already propagated, and (ii) $C_4$ is watched by $v_3$ and $v_2$.
\end{example}

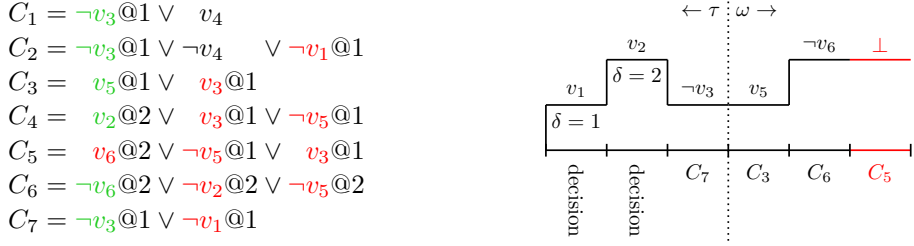
\begin{figure}[t]
  \centering
  \begin{subfigure}{0.35\textwidth}
  \centering
  \raisebox{\height}{
    \begin{tabular}{l}
        $C_{1} = \green{\neg v_{3}}\onLevel{1} \lor \phantom{\neg} v_{4}\phantom{\onLevel{1}}$\\
        $C_{2} = \green{\neg v_{3}}\onLevel{1}\lor \neg v_{4}\phantom{\onLevel{1}} \lor \red{\neg v_{1}}\onLevel{1}$\\
        $C_{3} = \green{\phantom{\neg} v_{5}}\onLevel{1} \lor \red{\phantom{\neg} v_{3}}\onLevel{1}$\\
        $C_{4} = \green{\phantom{\neg} v_{2}}\onLevel{2}\lor \red{\phantom{\neg} v_{3}}\onLevel{1}\lor \red{\neg v_{5}}\onLevel{1}$\\
        $C_{5} = \red{\phantom{\neg} v_{6}}\onLevel{2} \lor \red{\neg v_{5}}\onLevel{1} \lor \red{\phantom{\neg} v_{3}}\onLevel{1}$\\
        $C_{6} = \green{\neg v_{6}}\onLevel{2} \lor \red{\neg v_{2}}\onLevel{2} \lor \red{\neg v_{5}}\onLevel{2}$\\
        $C_{7} = \green{\neg v_{3}}\onLevel{1} \lor \red{\neg v_{1}}\onLevel{1}$\\
    \end{tabular}
  }
  \end{subfigure}
    \hfill 
    \begin{subfigure}{0.6\textwidth}
        \centering
        \tikz \node [scale=0.8] {
            \begin{tikzpicture}
                \input{Figures/trail-6.tex}
            \end{tikzpicture}
        };
    \end{subfigure}
    \caption{\(C_4\) is a MLI since it is satisfied only $v_2$ at level 2, and $v_3, \neg v_5$ are falsified at level 1. We use the notation \(v\onLevel{1}\) to indicate that literal \(v\) is on level~\(1\).
    }
    \label{fig:mli-ex}
\end{figure}

\section{Invariant Properties on CDCL Variants}
\label{sec:invariant-framework}
To properly handle MLI similar to \Cref{ex:MLI}, in this section
we revisit and formalize our invariants from \Cref{fig:invariant-list}, expressing properties that need to be maintained in (variants of) CDCL with chronological backtracking.

The crux of our invariant properties is captured by watched literals~\cite{DBLP:conf/dac/MoskewiczMZZM01}.
They reduce the number of clauses to be checked when propagating a literal.
\Cref{inv:weak-watched-literals} therefore expresses that, as long as CDCL does not falsify one of the watched literals $c_1,c_2$ of a clause $C$, the clause $C$ is not a conflict. Therefore, when propagating a watched literal $c_i$ during CDCL, only checking the clauses watched by $\neg c_i$ is sufficient to not miss any conflict.

\setcounter{invariant}{0}
\begin{invariant}[Weak watched literals]
    \label{inv:weak-watched-literals}
    Let \(\p = \trail \cdot \q\) be the current trail. For each clause $C \in F$ watched by the two distinct watched literals $c_1, c_2$, we have
    \(
        \neg c_1 \in \trail \Rightarrow \neg c_2 \notin \trail.
    \)
\end{invariant}

\Cref{inv:weak-watched-literals} ensures that conflicts are not missed during CDCL. Indeed, if there is a conflicting clause $C$, the conflict is found after propagating all literals of $C$. After propagation, no more literal has to be propagated, so \(\p = \trail\). A conflicting clause $C$ thus violates \Cref{inv:weak-watched-literals}, and hence the conflict of $C$ is captured during CDCL.

\subsection{CDCL Invariants on Implications}
\label{sec:convenient-cdcl-invariants}
Next, we ensure the soundness of unit implications. \Cref{inv:implied-literals} expresses that literals are either decisions or implied by a sound implication. Note that an implication can be performed if there is only one unassigned literal that can satisfy a clause $C$; hence, $C$ is a unit clause. In addition to ensuring that the solver infers correct literals, \Cref{inv:implied-literals} is also relevant for conflict analysis (see proof of \Cref{thm:learned-clause-sound}).
\begin{invariant}[Implied literals]
  \label{inv:implied-literals}
  If a literal $\ell$ is in the trail $\p$, then $\ell$ is either a decision literal or $\ell$ is implied by $\p$ and its reason $\reason(\ell)$. That is,
  \begin{equation*}
    \forall \ell \in \p.\ \ell \in \p^d ~\lor~ \left[\ell \in \reason(\ell) ~\land~ [\reason(\ell) \setminus \{\ell\}, \p] \vDash \bot\right].
  \end{equation*}
\end{invariant}

To perform conflict analysis with the first unique implication point (UIP) \cite{DBLP:conf/dac/MoskewiczMZZM01}, CDCL solving assumes that literals are organized in a topological sort of the implication graph.
\begin{invariant}[Topological order]
  \label{inv:topological-order}
  Trail $\p$ is a topological order of the implication graph:
  \begin{equation*}
    \forall \ell \in \p.\ \forall \ell' \in \reason(\ell) \setminus \{\ell\}.\ p_\p(\neg \ell') \leq p_\p(\ell),
  \end{equation*}
  where $p_\p(\ell)$ and $p_\p(\neg \ell')$ are respectively the positions of $\ell$ and $\neg \ell'$ in $\p$.
\end{invariant}

\noindent \Cref{inv:topological-order} holds by construction in CDCL with non-chronological backtracking (NCB) and chronological backtracking without reimplication. However, \Cref{inv:topological-order} is crucial in any setting of reimplying literals.

Finally, we impose that Boolean constraint propagation (BCP) in CDCL does not miss unit implication during proof search. \Cref{inv:strong-watched-literals} therefore formalizes that CDCL cannot have one propagated falsified watched literal without the clause being satisfied.
\begin{invariant}[Strong watched literals]
    \label{inv:strong-watched-literals}
    Consider the trail \(\p = \trail \cdot \q\). For each clause $C \in F$ watched by the two distinct watched literals $c_1, c_2$, we have
    \(
        \neg c_1 \in \trail \Rightarrow c_2 \in \p.
    \)
\end{invariant}

\Cref{inv:strong-watched-literals} strengthens \Cref{inv:weak-watched-literals}. When a conflicting clause $C$ is detected while propagating $\ell$, the literal $\ell$ cannot be added to $\trail$ without violating \Cref{inv:strong-watched-literals}. As such, by imposing \Cref{inv:strong-watched-literals}, the conflict of $C$ is resolved and the trail is adapted.

\subsection{Chronological Backtracking}
\label{sec:invariants-in-cb}
\Cref{inv:strong-watched-literals} holds for CDCL with NCB, since the trail contains monotonically increasing decision levels. Therefore, within NCB, literals are unassigned in the reverse order of propagation. In particular, if a literal $\ell$ is satisfied in a clause $C$ when propagating another literal $\ell'$, the literal $\ell$ remains satisfied at least until $\ell'$ is backtracked.

\paragraph{Weak chronological backtracking (WCB).}
When considering (variants of) chronological backtracking in CDCL, \Cref{inv:strong-watched-literals} becomes critical, as detailed next. The core idea is to save parts of the trail without repropagating unlike~\cite{DBLP:conf/sat/HickeyB20}.

\begin{example}\label{ex:MLI:CB}
Let us revisit the example of \Cref{fig:mli-ex}.
\Cref{fig:mli-ex-wcb} shows the trail after backtracking to level 1. Literal $v_3$ is already propagated ($v_3 \in \trail$), and $C_4$ is still watched by $v_2$ and $v_3$. Therefore, the implication of $v_2$ is missed, even though \Cref{inv:weak-watched-literals} is not violated.
\end{example}

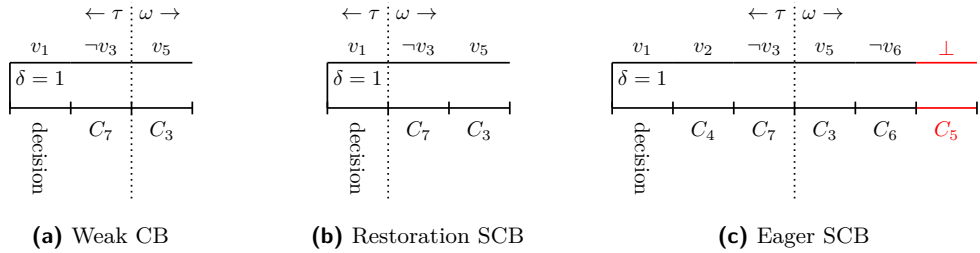
\begin{figure}[t]
  \centering
    \begin{subfigure}{0.29\textwidth}
        \centering
            \tikz \node [scale=0.8] {
                \begin{tikzpicture}
                    \input{Figures/trail-7-wcb.tex}
                \end{tikzpicture}
            };
            \captionsetup{justification=centering}
            \caption{Weak CB}
            \label{fig:mli-ex-wcb}
    \end{subfigure}
    \begin{subfigure}{0.29\textwidth}
        \centering
            \tikz \node [scale=0.8] {
                \begin{tikzpicture}
                    \input{Figures/trail-7-rscb}
                \end{tikzpicture}
            };
            \captionsetup{justification=centering}
            \caption{Restoration SCB}
            \label{fig:mli-ex-rscb}
    \end{subfigure}
    \begin{subfigure}{0.4\textwidth}
        \centering
            \tikz \node [scale=0.8] {
                \begin{tikzpicture}
                    \input{Figures/trail-6-escb.tex}
                \end{tikzpicture}
            };
            \captionsetup{justification=centering}
            \caption{Eager SCB}
            \label{fig:mli-ex-escb}
    \end{subfigure}
    \caption{Different CB ways of handling the missed lower implications of \Cref{fig:mli-ex}.
    }
    \label{fig:mli-ex-2}
\end{figure}

To circumvent the problem of missing implications similar to \Cref{ex:MLI:CB}, we distinguish a
\emph{weak chronological backtracking (WCB)} variant of
CDCL with chronological backtracking. Within WCB, \Cref{inv:strong-watched-literals} is not necessarily satisfied, as unit implications at lower levels can be missed. To recover \Cref{inv:strong-watched-literals} in variants of CDCL with CB, we adjust and label two existing solutions in SAT solving:
(i) \emph{restoration}~\cite{DBLP:conf/sat/NadelR18}, for repairing the trail $p$ after backtracking; and (ii) \emph{prophylaxis}~\cite{DBLP:conf/sat/Nadel22}, for forcing literals at the lowest possible level.

\paragraph{Restoration.}
We call \emph{restoration} the approach in which the trail $\p$ is repaired by pushing back the propagation head (the limit between $\trail$ and $\q$) when propagating~\cite{DBLP:conf/sat/NadelR18}. Out-of-order literals are repropagated whenever they are moved in the trail during backtracking. For example, in \Cref{fig:mli-ex-rscb}, $v_3$ was the first literal that changed position during backtracking, so this is where the propagation head is set. When backtracking to level $\level$, the propagation head is set to $p_\p(\p^d[\level])$. When $v_3$ is repropagated, $v_2$ is reimplied. We, therefore, restore \Cref{inv:strong-watched-literals} by repropagating the out-of-order literals.
We call this approach \emph{restoring strong chronological backtracking} (RSCB), allowing to restore the trail $\p$ by propagating more. It is also used in \cadical~\cite{BiereFleuryPolitt-SAT-Competition-2023-solvers}.

\paragraph{Prophylaxis.}
We name \emph{prophylaxis}\footnote{``Prophylaxis'' is a chess term referring to a move that deals with a threat before it becomes a problem.}
the approach in which missed lower implications are prevented from becoming missed unit implications~\cite{DBLP:conf/sat/Nadel22}. Prophylaxis
uses an eager reimplication procedure and imposes the validity of a compatibility invariant; we formalize this property in \Cref{inv:highest-watched-literals}. That is, when a clause $C$ is detected to be a missed lower implication of $\ell$, then $\ell$ is reimplied at level $\level(C\setminus\{\ell\})$ and its reason for propagation is updated.
Prophylaxis thus enforces our \emph{backtrack compatible} \Cref{inv:highest-watched-literals} by ensuring that no clause can become unit after backtracking. Furthermore, \Cref{inv:highest-watched-literals} guarantees that literals are always propagated at the lowest level, and conflicts are detected at the lowest level.

\begin{invariant}[Backward compatible watched literals]
    \label{inv:highest-watched-literals}
    For each clause $C \in F$ watched by the two distinct watched literals $c_1, c_2$, we have
    \(
        \neg c_1 \in \trail \Rightarrow \left[c_2 \in \p \land \level(c_2) \leq \level(c_1)\right].
    \)
\end{invariant}

\begin{example}
\Cref{fig:mli-ex-escb} shows the trail after $v_2$ is reimplied. In this case, $v_2$ was a decision, and $\neg v_6$ has to be reimplied to level 1 as well. All literals are propagated at the lowest possible level. Thus, using \Cref{inv:highest-watched-literals}, the conflict $C_5$ is properly detected at level 1, instead of level 2.
\Cref{fig:mli-ex-escb} also shows that the trail $\p$ no longer follows a topological order of the implication graph. These issues have to be addressed.
\end{example}

Based on \Cref{inv:highest-watched-literals}, \emph{eager strong chronological backtracking (ESCB)} is used in~\cite{DBLP:conf/sat/Nadel22,cadical-reimply}, yielding a CDCL method with chronological backtracking that satisfies \Cref{inv:highest-watched-literals} by eager reimplication of missed lower implications. In \Cref{tab:CB-variants-in-implementation} we summarize backtracking strategies in CDCL, also listing our solution in this respect: \emph{lazy reimplication in strong chronological backtracking (LSCB)}. Our LSCB approach maintains \Cref{inv:weak-watched-literals} and \Cref{inv:strong-watched-literals}, while weakening \Cref{inv:highest-watched-literals} via \Cref{inv:lazy-highest-watched-literals}, as described next in \Cref{sec:practical-approach} and implemented in \Cref{alg:cdcl}.

\begin{table}[t]
    \centering
    \caption{\label{tab:CB-variants-in-implementation}CB variants in CDCL, together with their invariant properties.}
    \begin{tabular}{l|c c c c | l}
        & Inv.~\ref{inv:weak-watched-literals}
        & Inv.~\ref{inv:strong-watched-literals}
        & Inv.~\ref{inv:highest-watched-literals}
        & Inv.~\ref{inv:lazy-highest-watched-literals}
        & Solvers \\
        \hline
        NCB  & \cmark & \cmark & \cmark & \cmark
        & Most CDCL solvers\\
        WCB  & \cmark & \xmark & \xmark & \xmark
        & Our work -- \napsat\\
        RSCB & \cmark & \cmark & \xmark & \xmark
        & \textsc{Maple\_LCM\_Dist}~\cite{DBLP:conf/sat/NadelR18}, \cadical \\
        ESCB & \cmark & \cmark & \cmark & \cmark
        & IntelSAT and \cadical~1.9.4 \\
        LSCB & \cmark & \cmark & \xmark & \cmark
        & Our work -- \napsat, now in \cadical\\
    \end{tabular}
\end{table}

\section{Adapting CDCL with Lazy Reimplications}
\label{sec:practical-approach}
Embedding the prophylaxis approach of \Cref{sec:invariants-in-cb} in existing
CDCL data structures is highly non-trivial, due to the rigid and entangled data structures~\cite{DBLP:conf/sat/Nadel22,cadical-reimply}, see e.g.~\cite{coutelier2023chronological}. In addition, reimplying literals~\cite{DBLP:conf/sat/Nadel22,cadical-reimply} changes the implication graph, and hence the trail $\p$ is no longer a topological sort of the implications; as such, \Cref{inv:topological-order} must be restored.

While the restoration approach of \Cref{sec:invariants-in-cb} offers a practically simpler solution, restoration might require the re-propagation of a large part of $\p$ and thus can be computationally very expensive. For example, while in \Cref{fig:mli-ex-rscb} only one literal had to be re-propagated, re-propagation could be applied on an arbitrary number of literals.

\paragraph{Our solution: Lazy reimplication in CDCL.}
To overcome inefficiencies of restoration and pure prophylaxis, our work advocates a \emph{lazy reimplication} technique for CDCL with strong chronological backtracking. To ensure \Cref{inv:strong-watched-literals}, we reimply literals after backtracking. That is, \emph{we detect missed lower implications eagerly but reimply them lazily}.

Our lazy reimplication approach for CDCL-based solving is summarized in \Cref{alg:cdcl}.
In what follows, we describe the key ingredients of \Cref{alg:cdcl} and revise the CDCL invariants of \Cref{sec:invariant-framework}, adjusted to \Cref{alg:cdcl}. To this end, we introduce a lazy reimplication vector $\lazy$ to store missed lower implications, where $\lazy$ is a function from literals to clauses. Intuitively, the lazy reimplication vector $\lazy$ stores the lowest detected missed lower implication for each literal $\ell$.
The clause $\lazy(\ell) \neq \square$ is an alternative reason that would propagate $\ell$ in trail $\p$, lower than the reason $\reason(\ell)$.
Initially, no clause is assigned, and $\forall \ell.\ \lazy(\ell) = \square$ (that is, the undefined clause). \Cref{inv:lazy-reimplication-correctness} is asserted to hold during proof search.

\begin{invariant}[Lazy reimplication]
  \label{inv:lazy-reimplication-correctness}
  If the lazy reimplication reason $\lazy(\ell)$ of literal $\ell$ is defined, then the clause $\lazy(\ell)$ is a missed lower implication of $\ell$. That is,
  \begin{align*}
    \lazy(\ell) \neq \square
      ~\Rightarrow\ & ~~~~~\ell \in \p
      \land\  \ell \in \lazy(\ell) \\
       &  \mathrel\land \big(\lazy(\ell) \setminus \{\ell\} \land \p\big) \vDash \bot\\
       &  \mathrel\land \level(\lazy(\ell) \setminus \{\ell\}) < \level(\ell)
  \end{align*}
\end{invariant}

When a missed lower implication for $\ell$ is detected, then $\ell$ is not reimplied directly. Rather, we store the MLI in $\lazy$ until $\ell$ is unassigned during backtracking. For example, if a literal $\ell$ is assigned at level 3 and a missed lower implication $C$ for $\ell$ is detected with $\level(C \setminus \{\ell\}) = 1$, then backtracking to level 2 will reassign $\ell$ from level 3 to level 1 by $C$.

Using our lazy reimplication vector $\lazy$, we weaken \Cref{inv:highest-watched-literals} into \Cref{inv:lazy-highest-watched-literals} such that, during backtracking, we identify missed lower implications without requiring the re-propagation of out-of-order literals.

\begin{invariant}[Lazy backtrack compatible watched literals]
  \label{inv:lazy-highest-watched-literals}
  Consider the trail \(\p = \trail \cdot \q\). For each clause $C \in F$, if one watched literal $c_1$ of $C$ is falsified by $\trail$, then the other $c_2$ must be satisfied at a lower level, or a missed lower implication lower than $c_1$ is set in $\lazy$.
  \begin{equation*}
    \neg c_1 \in \trail \Rightarrow \Big(c_2 \in \p \land \big( \level(c_2) \leq \level(c_1) \lor \level(\lazy(c_2) \setminus \{c_2\}) \leq \level(c_1)\big)
    \Big)
    \label{eq:watched-literals}
  \end{equation*}
\end{invariant}

\paragraph{Lazy reimplication for strong chronological backtracking -- LSCB.} Guided by the reimplication and backtracking properties of \Cref{inv:lazy-reimplication-correctness} and \Cref{inv:lazy-highest-watched-literals},
\Cref{alg:cdcl} shows our LSCB algorithm for CDCL with chronological backtracking, as a slight refactoring of weak chronological backtracking (WCB). In the following algorithms, particularities of LSCB are highlighted in\blue{blue}.

\begin{algorithm}[t]
  \caption{Lazy Reimplication in CDCL with CB}
  \label{alg:cdcl}
  \input{Algorithms/cdcl.tex}
\end{algorithm}

An important detail should be noted upon \Cref{alg:cdcl}:
in our abstract representation, it is not explicitly checked whether the learned clause $D$ is different from the conflicting clause $C$; such a check, however, should be performed when implementing \Cref{alg:cdcl}. Indeed, as pointed out in RSCB~\cite{DBLP:conf/sat/MohleB19}, it is possible that a conflicting clause $C$ does not require conflict analysis since $C$ might already be a UIP.
However, if the highest literal $\ell$ in $C$ is a MLI, then the clause might be conflicting again after backtracking (see \Cref{alg:conflict-analysis}).

\begin{example}
Consider the example of \Cref{fig:mli-ex}. Here, the conflicting clause $C_5$ only has one literal at the highest level, and, as such, it qualifies as a UIP. Therefore no conflict analysis is required, we only backtrack to level 1, and then $C_5$ implies $v_6$ at level 1. However, if $\neg v_6$ was a missed lower implication, then backtracking to level 1 would reimply $\neg v_6$, with $C_5$ conflicting again; this time, however, $C_5$ would require conflict analysis.
\end{example}

\begin{algorithm}[t]
  \caption{Boolean Constraint Propagation in LSCB}
  \label{alg:bcp}
  \input{Algorithms/propagate-lit.tex}
  \input{Algorithms/bcp.tex}
\end{algorithm}
\paragraph{Propagation in LSCB.}
When falsifying a watched literal, \Cref{alg:cdcl} might need to find a replacement candidate to become the new watched literal (line~7 of \Cref{alg:bcp}). We define the property of the candidate literal with information about its level as below.
\newpage
\begin{definition}[Candidate literal]
  \label{def:replacement-return-value}
  Let clause $C$ be watched by the literals $c_1$ and $c_2$. with $\neg c_1 \in \q$. Then, \Call{SearchReplacement}{$C, c_1, c_2$} from \Cref{alg:bcp} returns a \emph{candidate literal} $r$ for which one of the following holds:
  \begin{itemize}
    \item  \Cref{inv:highest-watched-literals} is satisfied on $C$ after $\neg c_1$ is added to $\trail$, i.e.\\
    \(
      \neg r \in (\trail \cdot \neg c_1) \Rightarrow c_2 \in \p \land \level(c_2) \leq \level(r);
    \)
    \item \(C\) is conflicting, propagating, or a MLI for \(c_2\). As such, $C \setminus \{c_2\}$ is unsatisfiable with the current assignment, and $r$ is at the highest decision level in $C \setminus \{c_2\}$, that is\\
    \(
      \big(C \setminus \{c_2\} \land \p\big) \vDash \bot ~~\land~~ \level(r) = \level(C \setminus \{c_2\})
    \)
  \end{itemize}
\end{definition}

Concretely, the \Call{SearchReplacement}{$C, c_1, c_2$} procedure iterates over literals of $C\setminus \{c_2\}$ and stops when it finds a literal $r$ that would satisfy \Cref{inv:highest-watched-literals} if $c_1$ was replaced by $r$. In case of failure, it returns the highest literal in $C\setminus \{c_2\}$.
The knowledge of the highest literal in $C\setminus\{c_2\}$ is enough to determine the nature and level of the clause.

\Cref{alg:bcp} shows our Boolean constraint propagation (BCP) algorithm adapted to support LSCB. As opposed to standard BCP, \Cref{alg:bcp} does not stop when the other watched literal is satisfied. We need the extra guarantee that either $c_2$ is implied at a level lower than $c_1$, or it is registered as a MLI before skipping the clause.
Further, when a non-falsified replacement literal cannot be found,
\Cref{alg:bcp} still changes the watched literal.
While this is not always strictly necessary (for example, in conflicts), systematically swapping the highest literal allows checking the level of the clause in constant time and provides cheap useful properties to the clause.

\paragraph{Backtracking in LSCB.}
When backtracking, our LSCB approach has the information of whether a clause $C$ violates \Cref{inv:strong-watched-literals}. Therefore, \Cref{alg:backtracking} can directly imply those missed lower implications (line~15 of \Cref{alg:backtracking}).

The order in which literals are reimplied in \Cref{alg:backtracking} is not important, as shown later in \Cref{thm:reimplied-literal-topological-order}. It is, however, unclear whether a specific order would impact performance in problems where the stability of literal position in the trail is important. In such cases, ordering the reimplications in increasing levels might be beneficial.

\begin{algorithm}[t]
  \caption{Backtracking and Reimplication.}
  \label{alg:backtracking}
  \input{Algorithms/backtracking.tex}
\end{algorithm}

\paragraph{Conflict analysis with MLI.}
As opposed to traditional backtracking, \Cref{alg:cdcl} does not guarantee that, once it backtracks to a level lower than the level of the learned clause $D$, the clause $D$ will be propagating. Indeed, let the falsified learned clause $D = \{c_1, c_2, \dots, c_m\}$ with $c_1$ a unique literal at level $\level(D)$. If we backtrack to level $\level(D) - 1$, $c_1$ might be reimplied at a lower level, and $D$ would still be a conflict. In response to this, we propose the following two solutions:
\begin{enumerate}[({Analyze}-1)]
  \item \label{strat:reanalyze} we analyze the conflict and backtrack again until we get a unit clause;
  \item \label{strat:analyze-resolve-UIP} we perform conflict analysis with the knowledge of missed lower implications.
\end{enumerate}
In \Cref{alg:conflict-analysis} we chose option \ref{strat:analyze-resolve-UIP}.
Option \ref{strat:reanalyze} will generate the same clause in the end, but might create some unnecessary ones in the process. We empirically check our intuition in \Cref{sec:empirical-results} and demonstrate that option \ref{strat:analyze-resolve-UIP} indeed works better.

We refer to $D\otimes_\ell C'$ as the result of binary resolution applied to the clauses $C$ and $D$ over the literal $\ell$.

In \Cref{alg:conflict-analysis}, when possible, we use the lazy reimplication reason $\lazy(\ell)$ instead of the real reason $\reason(\ell)$ during conflict analysis. The lazy reason $\lazy(\ell)$ is guaranteed to introduce literals at a level lower than $\level(C)$, making it converge to a UIP faster.
Once a UIP is obtained, \Cref{alg:conflict-analysis} does not stop if there exists a missed lower implication for the last literal at the conflict level. Furthermore, we adapted the learnt clause minimization approach~\cite{DBLP:conf/sat/SorenssonB09}, adjusted to
\Cref{alg:conflict-analysis} so that both reasons are checked if the literal can be removed.

\begin{algorithm}[t]
  \caption{Conflict Analysis.}
  \label{alg:conflict-analysis}
  \input{Algorithms/conflict-analysis.tex}
\end{algorithm}

\begin{example}
    \Cref{fig:mli-in-analysis} shows a conflict after \Cref{alg:bcp} detected a missed lower implication $C_6$. From \Cref{inv:lazy-highest-watched-literals}, we have $\lazy(\neg v_3) = C_6$. \Cref{alg:cdcl} will then trigger \Cref{alg:conflict-analysis} to analyse the conflict on $C_5$. During conflict analysis with \Cref{alg:conflict-analysis}, we start from the conflicting clause
$D = \red{\neg v_7}\onLevel{2}
\lor \red{v_5}\onLevel{2}
\lor \red{v_6}\onLevel{1}$
and apply the resolution $D \gets D \otimes_{\neg v_7} C_4$ to obtain
$D = \red{v_5}\onLevel{2}
\lor \red{v_3}\onLevel{2}
\lor \red{v_6}\onLevel{1}$.
We once again apply resolution and have $D \gets D \otimes_{v_5} C_2$, yielding the clause
$D = \red{v_3}\onLevel{2}
\lor \red{v_6}\onLevel{1}
\lor \red{\neg v_4}\onLevel{1}$.
As this $D$ is a UIP, most CDCL approaches would stop conflict analysis here. However, in our LSCB approach we know that $v_3$ can be reimplied at level 1. Therefore, after backtracking to level 1 and reimplying $\neg v_3$ with \Cref{alg:backtracking}, the clause $D$ would still be conflicting and conflict analysis would need to be triggered again. Instead we apply the resolution $D \gets D \otimes_{v_3} C_6$ to get a clause at level 1, namely clause
$D = \red{v_6}\onLevel{1}
\lor \red{\neg v_4}\onLevel{1}
\lor \red{v_2}\onLevel{1}$.
We then continue until the procedure at level 1 and obtain the final clause $D = \red{v_2} \onLevel{1}$.
\end{example}
\begin{figure}[t]
\centering
  \begin{subfigure}{0.35\textwidth}
  \centering
  \raisebox{\height}{
    \begin{tabular}{l}
        $C_{1} = \green{\neg v_{2}}\onLevel{1}
            \lor \red{\phantom{\neg} v_{1}}\onLevel{1}$\\
        $C_{2} = \green{\neg v_5}\onLevel{2}
            \lor \red{\phantom{\neg} v_3}\onLevel{2}
            \lor \red{\neg v_4}\onLevel{1}$\\
        $C_{3} = \green{\neg v_6}\onLevel{1}
            \lor \red{\phantom{\neg} v_2}\onLevel{1}
            \lor \red{\neg v_4}\onLevel{1}$\\
        $C_{4} = \green{\phantom{\neg} v_7}\onLevel{2}
            \lor \red{\phantom{\neg} v_5}\onLevel{2}
            \lor \red{\phantom{\neg} v_3}\onLevel{1}$\\
        $C_{5} = \red{\phantom{\neg} v_5}\onLevel{2}
            \lor \red{\neg v_7}\onLevel{2}
            \lor \red{\phantom{\neg} v_6}\onLevel{1}$\\
        $C_{6} = \green{\neg v_3}\onLevel{2}
            \lor \red{\neg v_4}\onLevel{1}
            \lor \red{\phantom{\neg} v_2}\onLevel{1}$\\
        $C_{7} = \green{\phantom{\neg} v_{4}}\onLevel{1}
            \lor \red{\phantom{\neg} v_{2}}\onLevel{1}$\\
    \end{tabular}
  }
  \end{subfigure}
    \hfill 
    \begin{subfigure}{0.6\textwidth}
        \centering
        \tikz \node [scale=0.8] {
                \begin{tikzpicture}
                    \input{Figures/trail-analysis}
                \end{tikzpicture}
            };
    \end{subfigure}
    \caption{The clause $\neg v_{3}\onLevel{2} \lor \neg v_4\onLevel{1} \lor v_2\onLevel{1}$ is a missed lower implication in this example. $v_2$ and $\neg v_4$ are falsified at level 1, whereas $\neg v_{3}$ is only satisfied at level 2.}
    \label{fig:mli-in-analysis}
\end{figure}
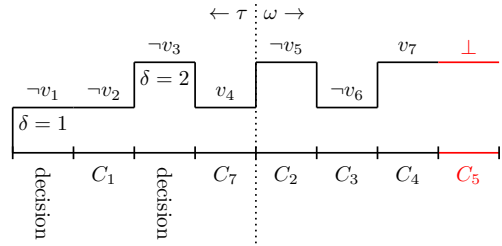
\newpage

\section{Soundness of Lazy Reimplication}
\label{sec:soundness}
This section proves the soundness and completeness\footnote{with details also in the code base of \napsat} of our LSCB approach given in \Cref{alg:cdcl}. We note that \Cref{alg:cdcl} implements strong chronological backtracking and does not miss any implication; as such, \Cref{inv:strong-watched-literals} holds.

\begin{theorem}[Soundness of conflict analysis]
    \label{thm:learned-clause-sound}
    Let $C \in F$ be a conflicting clause with the partial assignment $\p$. Then, conflict analysis in $\Call{Analyze}{C}$ from \Cref{alg:cdcl} returns a conflicting clause that is implied by the clause set.
\end{theorem}
\begin{proof}
    The starting clause $D \gets C$ is conflicting. At each step, $D$ is resolved with a clause $C'$ such that $C' = \reason(\ell)$ or $C' = \lazy(\ell)$, with $\ell \in C'$ and $\neg \ell \in D$. From the definition of $\reason$ and $\lazy$, we have $\big(C' \setminus \{\ell\}\land \p\big) \vDash \bot$. Therefore, the clause $D \gets D \otimes_{\neg\ell} C'$ is conflicting, and implied by F, since $C' \in F$.
\end{proof}

\begin{theorem}[No missed unit implication]
    \label{thm:lazy-highest}
    \Cref{alg:cdcl} satisfies \Cref{inv:lazy-highest-watched-literals}. As such, our LSCB method in \Cref{alg:cdcl} does not miss unit implications.
\end{theorem}
\newpage
\begin{proof}
    We prove that \Cref{inv:lazy-highest-watched-literals} holds for each building block of \Cref{alg:cdcl}.
    \paragraph{BCP.}
    \Cref{inv:lazy-highest-watched-literals} trivially holds at the starting state, where $\p = \emptyset$.
    Further, during the propagation of one literal, \Cref{alg:bcp} ensures that for each clause $C \in F$ watched by $c_1$ and $c_2$, the following Hoare triple holds
    $$\{P\} \Call{PropagateLiterall}{\ell}\{Q\},$$
    where
    \begin{align*}
        P \equiv \neg c_1 \in \trail \Rightarrow \left[c_2 \in \p \land \left[ \level(c_2) \leq \level(c_1) \lor \level(\lazy(c_2) \setminus \{c_2\}) \leq \level(c_1)\right]\right]\\
        Q \equiv \neg c_1 \in (\trail \cdot \ell) \Rightarrow \left[c_2 \in \p \land \left[ \level(c_2) \leq \level(c_1) \lor \level(\lazy(c_2) \setminus \{c_2\}) \leq \level(c_1)\right]\right]
    \end{align*}
    By structural induction over the statements of \Cref{alg:bcp}, we conclude that \Cref{inv:lazy-highest-watched-literals} is maintained by BCP.

    \paragraph{Backtracking.}
    During backtracking in \Cref{alg:backtracking}, each literal $c_i$ is inspected: $c_i$ is either removed from the trail $\p$ or $c_i$ is kept. Violating \Cref{inv:lazy-highest-watched-literals} means that a literal $c_2$ from the trail is removed such that $\neg c_1 \in \trail \land c_2 \notin \p$ for some clause $C = \{c_1, c_2, \dots, c_m\}$ (since the levels are not altered). However, this case is rectified, since either $\level(c_1) \leq \level(c_2)$ (and then $\neg c_1$ would be removed from $\trail$), or $\level(\lazy(c_2) \setminus \{c_2\}) \leq \level(c_1)$ (and then $c_2$ would be reimplied at level $\level(\lazy(c_2) \setminus \{c_2\})$ and $c_2 \in \p \land \level(c_2) \leq \level(c_1)$ would be true), or $\level(\lazy(c_2) \setminus \{c_2\}) > \level(c_1)$ (and then $\neg c_1$ is also backtracked). As such, backtracking in \Cref{alg:backtracking} preserves \Cref{inv:lazy-highest-watched-literals}.

    \paragraph{Analysis.}
    Within conflict analysis in \Cref{alg:conflict-analysis}, the state of CDCL is not modified, only read. Therefore, any invariant that held before \Cref{alg:conflict-analysis} also holds after \Cref{alg:conflict-analysis}.

    \paragraph{CDCL.} We finally ensure that
    \Cref{inv:lazy-highest-watched-literals} is maintained by \Cref{alg:cdcl} also during its decision step and while adding a clause to the formula $F$.
    First, deciding still preserves \Cref{inv:lazy-highest-watched-literals}, since it merely adds a non-assigned literal to the propagation queue $\q$. Second, after backtracking in \Cref{alg:cdcl}, we know by construction that the learned clause will have a single literal $\ell$ that is unassigned. This literal $\ell$ is then implied at level $\level(D\setminus\{\ell\})$, satisfying \Cref{inv:lazy-highest-watched-literals} since the second watched literal $c_2$ is falsified at level $\level(D\setminus\{\ell\})$.
\end{proof}

\begin{corollary}[No missed conflict/implication]
    \label{cor:-strong-watched-literals}
    Our LSCB method from \Cref{alg:cdcl} preserves the strong watched literal property of \Cref{inv:strong-watched-literals}.
\end{corollary}

Based on the results above, we conclude the soundness and completeness of LSCB.

\begin{theorem}[LSCB soundness and completeness]
    Lazy reimplication with strong chronological backtracking from \Cref{alg:cdcl} is sound and complete.
\end{theorem}
\begin{proof}
    \Cref{thm:learned-clause-sound} implies that clauses added to the clause set are implied by $F$. By induction, if $\phi$ is the original CNF, then if $\phi \models F$ and $F \models C$, then $\phi \models F \cup \{C\}$.
    Furthermore, from \Cref{cor:-strong-watched-literals} we conclude that \Cref{inv:strong-watched-literals} holds.

    \Cref{alg:cdcl} returns unsat iff there exists a conflict at level 0; that is, there exists a set of clauses $F' \subseteq F$ such that $F' \models \bot$. As $\phi \models F$, then $\phi \models \bot$, and thus $\phi$ is unsatisfiable.
    Otherwise, \Cref{alg:cdcl} returns SAT if a model $\trail$ exists such that every variable has been assigned and propagated ($\p = \trail$). Based on \Cref{inv:strong-watched-literals},
    no conflict is possible and $\phi \models \trail$.
\end{proof}

\begin{theorem}[Topological order in LSCB]
    \label{thm:reimplied-literal-topological-order}
    The literals reimplied by the backtracking procedure of \Cref{alg:backtracking} respect the topological order of the implication graph.
\end{theorem}
\begin{proof}
    The reimplied literals cannot depend on each other. Indeed, if they are reimplied, their implication level before backtracking was higher than $d$. Therefore, if a literal $\ell$ depends on a literal $\ell'$ in the implication graph, then $\level(\ell) \geq \level(\ell')$. If the missed lower implication $\lazy(\ell)$ has a level lower than $d$, then all literals in $\lazy(\ell) \setminus\{\ell\}$ are lower than $d$, and therefore were not backtracked. Therefore, since all literals are independent, they can be reimplied in an arbitrary order at the end of the trail, and still respect the topological order.
\end{proof}

\section{Empirical Analysis}
\label{sec:empirical-results}
In this section, we discuss the implementation of \Cref{alg:cdcl} in our new SAT solver \napsat.
We also integrated it in \cadical and \glucose, and present our empirical results using \napsat, \cadical, and \glucose.

\subsection{\napsat for Lazy Reimplication in CDCL}
We implemented our LSCB method from \Cref{alg:cdcl} in the new SAT solver \napsat. Our \napsat tool  is a CDCL solver using the watcher list scheme~\cite{DBLP:conf/dac/MoskewiczMZZM01} with blocker literals~\cite{DBLP:journals/jsat/ChuHS09}. \napsat supports the backtracking variants of NCB, WCB, RSCB, and LSCB at runtime. In chronological backtracking, the backtracking scheme is purely chronological, that is, \napsat always backtracks to one level before the conflict (unlike \cadical).
\napsat uses the VSIDS decision heuristic~\cite{DBLP:conf/dac/MoskewiczMZZM01} with the agility restart strategy~\cite{DBLP:conf/sat/Biere08} and root-level clause elimination~\cite{BiereJarvisaloKiesl-SAT-Handbook-2021}.
\napsat is available at
\url{https://github.com/RobCoutel/NapSAT}
and consists in a total of \(\sim\)\SI{5,800}{loc}, among which the core of the solver represents \(\sim\)\SI{1,500}{\loc}.

\paragraph{Blocker literals in \napsat.}
Blocker literals are useful to reduce the number of pointer dereferencing of the literal pointer~\cite{DBLP:journals/jsat/ChuHS09}.
If the blocker $b$ is assigned at a level higher than the literal $\ell$ being falsified, then it might get backtracked before $\ell$ and a conflict might be missed. \Cref{inv:lazy-highest-watched-literals} can therefore be weakened, while still ensuring that no unit implication is missed.
\begin{invariant}[Lazy backtrack compatible watched literals with blocker literals]
  For each clause $C \in F$ watched by the two distinct literals $c_1, c_2$ and with blocker $b$, we have
    \begin{align*}
        \neg c_1 \in \trail \Rightarrow \Big(c_2 \in \p \land \big( \level(c_2) \leq \level(c_1) \lor \level(\lazy(c_2) \setminus \{c_2\}) \leq \level(c_1)\big)\Big) \\
        \lor \Big(b \in \p \land \level(b) \leq \level(c_1)\Big)
    \end{align*}
\end{invariant}
\noindent The eager update of blocking literals is in essence similar to strategies that aggressively update watched literals during BCP~\cite{watchsat}.

\paragraph{Experiments.}
\Cref{fig:comparison-backtrack-method} shows the average total number of propagations of \napsat on the 3-SAT uniform random problems from SATLIB~\cite{hoos2000satlib}. Our LSCB method from \Cref{alg:cdcl}, indicated via \texttt{-lscb}, performs better than the other backtracking versions of \napsat, both for satisfiable and unsatisfiable instances.
\Cref{fig:scatter} shows more details. In particular, it shows the total number of propagations of each unsatisfiable problem with 250 variables. It shows that LSCB consistently has fewer propagations than NCB, WCB, and RSCB.

We acknowledge that the number of propagations alone is not always representative of the real performance of a SAT solver, since propagation in LSCB is slightly more expensive than in NCB or WCB.
However, the number of propagations in \napsat indicates the impact of missed lower implications. For example, \Cref{fig:comparison-backtrack-method} shows that restoring the trail with RSCB might not be worth finding the missed lower implications in the random 3-SAT benchmarks; yet, reimplying literals lazily significantly reduces the total number of propagations.

\begin{figure}[t]
    \centering
    \subfloat[Satisfiable instances]{
        \includegraphics[width=.49\textwidth]{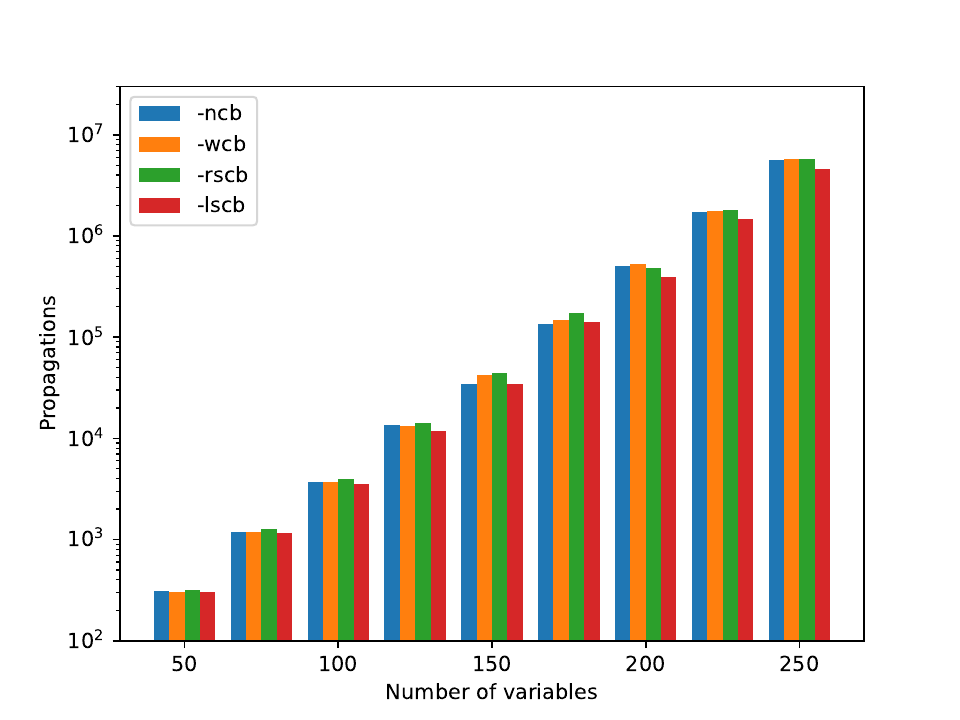}
    }
    \subfloat[Unsatisfiable instances]{
        \includegraphics[width=.49\textwidth]{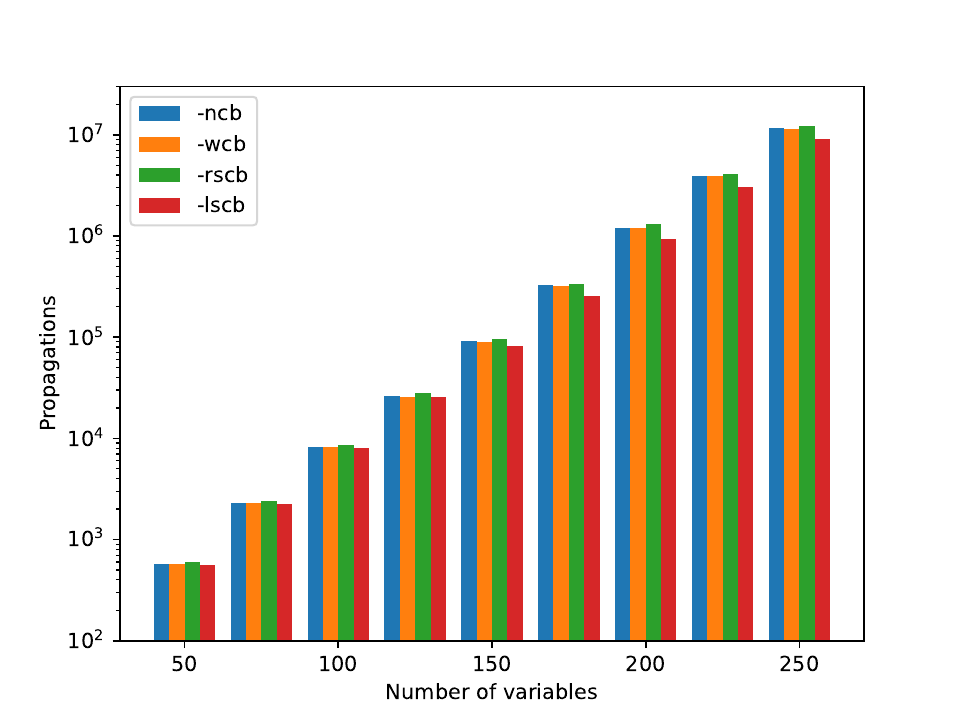}
    }
    \caption{Average total number of propagations performed by \napsat on the SATLIB 3-SAT random problem, clustered by the number of variables, and backtracking technique employed.}
    \label{fig:comparison-backtrack-method}
\end{figure}

\begin{figure}[t]
    \centering
    \subfloat[LSCB vs. NCB]{
        \includegraphics[width=.32\textwidth]{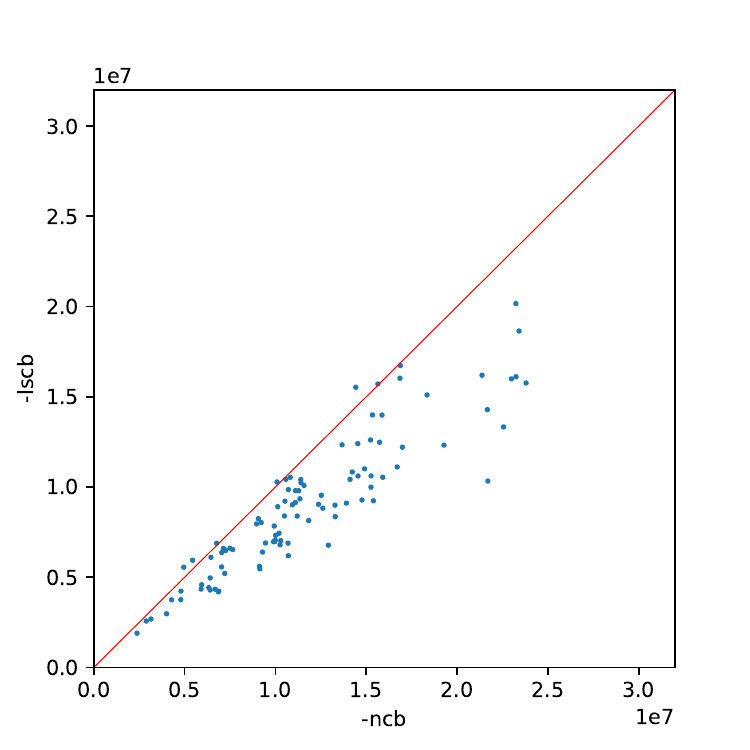}
    }
    \subfloat[LSCB vs. WCB]{
        \includegraphics[width=.32\textwidth]{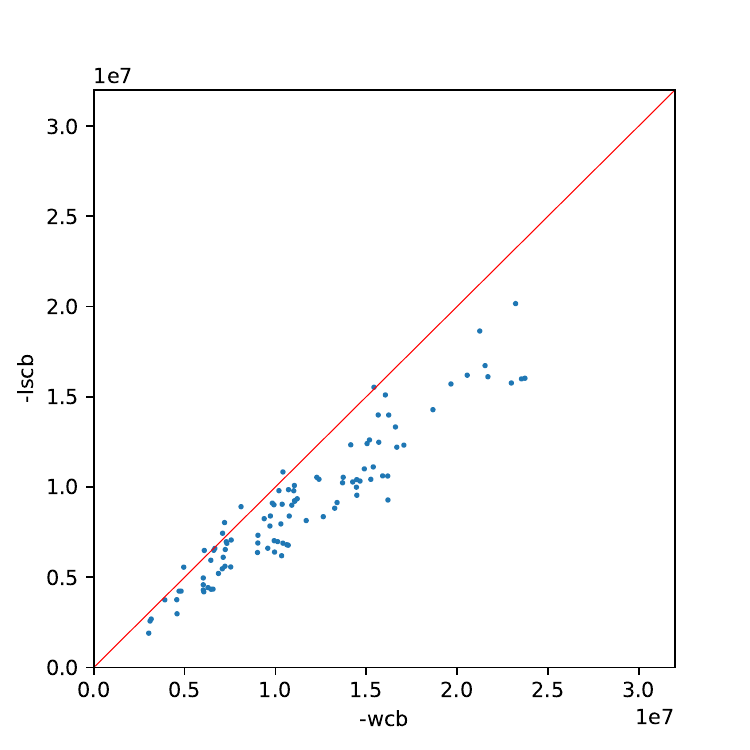}
    }
    \subfloat[LSCB vs. RSCB]{
        \includegraphics[width=.32\textwidth]{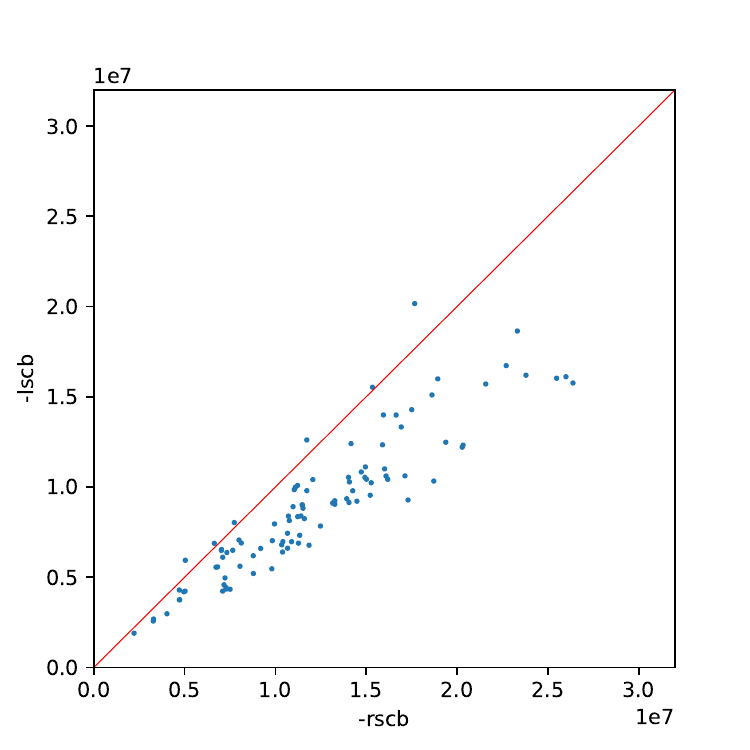}
    }
    \caption{Total number of propagations of \napsat for each unsatisfiable Uniform Random 3-SAT problem from SATLIB. The red line is the equality line. Marks under the equality line favour our new approach.}
    \label{fig:scatter}
\end{figure}

\subsection{Integration of LSCB in \cadical and \glucose}
\paragraph{LSCB in \cadical.} We implemented our LSCB approach from \Cref{alg:cdcl} in \cadical~\cite{cadical}, the baseline solver of the hack track of the SAT Competition. Thanks to the built-in model-based tester \textsc{Mobical}, the most effort came with ensuring that we have implemented correctly \Cref{inv:strong-watched-literals}: \cadical does not require watching the literals of two highest levels when the clause is propagating. This, however, requires iterating over the clause to find the propagation level, which we do not need.

We remark that we did not change the default backjumping policy of \cadical: when backjumping from more than \num{100} levels occurs (following the value implemented in \cadical), we resort to backtracking (going one level back). Otherwise, an algorithm similar to trail reuse for restarts is used to decide how many levels should be kept. Unlike the version implemented in \napsat, in \cadical we store the missed level instead of checking the level of the MLI each time we need the level.

We tested various configurations, as summarized in \Cref{tab:sc2023-cadical}, on the \textsf{bwForCluster} Helix with AMD Milan EPYC 7513 CPUs,
using a memory limit of \SI{16}{\giga\byte} RAM on the problems from the SAT Competition. Overall, we can see there is little difference between the considered configurations. In particular, the performance difference between WCB and NCB is limited, making it unclear if chronological backtracking is important. However, similar to the original \cadical implementation~\cite{DBLP:conf/sat/MohleB19}, on the benchmarks from the SAT Competition 2018, there is an improvement from WCB over NCB.
Our intuition is that chronological backtracking is especially useful when the decision heuristic is picking the wrong literals finding conflicts late instead of early (like finding a new unit at level 500 instead of level 1). The decision heuristics seem to perform worse on the 2018 benchmarks, while this did not seem to have happened since.

While the results in \napsat seem to indicate a large decrease in the number of propagations, three factors mitigate this effect in \cadical, as follows.
(i) Propagating a literal $\ell$ a second time as in RSCB is cheaper than propagating it for the first time. Most clauses remaining in the watcher list of $\ell$ will already be satisfied and are faster to check.
(ii) RSCB allows to use of blocking literals more loosely. There is no need to compare the level of the blocking literal and the propagated one, making them more potent.
(iii) Searching for a replacement literal is slightly more expensive in LSCB since we need to record the highest literal in the clause.

In a context where propagations are more expensive, such as SMT or user-propagators \cite{DBLP:conf/smt/BjornerEK22}, these considerations might weigh less on the overall performance of the solver. We will investigate these applications for future work.

\paragraph{LSCB in \glucose.} We also implemented our \Cref{alg:cdcl} for LSCB into the latest version of \glucose~\cite{glucose}, the SAT solver that pioneered the LBD heuristic for the usefulness of clauses (only without the minimization part). This is the only solver where we implemented the LSCB without any existing CB in the code. The entire diff (including new logging information and more assertions) is less than \num{1,000} lines. Our actual implementation of LSCB in \glucose is very close to our abstract \Cref{alg:cdcl}, because the blocker literal is always exactly the other watched literal. We use the simple heuristic to backtrack one level if jumping back more than 100, otherwise use the normal backjumping. We did not change the heuristic to block restarts~\cite{DBLP:journals/ijait/AudemardS18}, which is based on the trail length.

While running \glucose with LSCB on the SAT Competition 2023 (Fig.~\ref{fig:cdf-glucose-2023}), we observed worse performance. Interestingly, this is mostly due to one family of benchmarks, \texttt{SC23\_Timetable}, that perform much worse with strong backtracking (but are solved eventually). On the 2018 benchmarks again (Fig.~\ref{fig:cdf-glucose-2018}), we observed a slight performance improvement when using \glucose with LSCB and it seems to be better
to trigger chronological backtracking more often.

\begin{table}[tb]
    \centering
    \caption{\label{tab:sc2023-cadical}Number of solved instances by different variants of strong backtracking on the SC2023 competition, using a \SI{5,000}{\second} timeout}
    \renewcommand{\arraystretch}{1.2}
    \begin{tabular}{p{7cm}r@{\quad}r}
        \cadical version& solved&PAR-2 (\(\times 10^3\)) \\\midrule
        base-\cadical = RSCB & \textbf{248} & \textbf{4.09}\\
        LSCB, Analyze-\ref{strat:analyze-resolve-UIP} and minimization & 246& 4.16\\
        ESCB & 245& 4.16\\
        LSCB and Analyze-\ref{strat:analyze-resolve-UIP} & 246& 4.19\\
        NCB & 247& 4.19\\
        LSCB and Analyze-\ref{strat:reanalyze} & 242 & 4.24\\
    \end{tabular}
\end{table}

\begin{figure}[h!]
    \centering
    \begin{subfigure}{1\textwidth}
    \includegraphics[width=\textwidth]{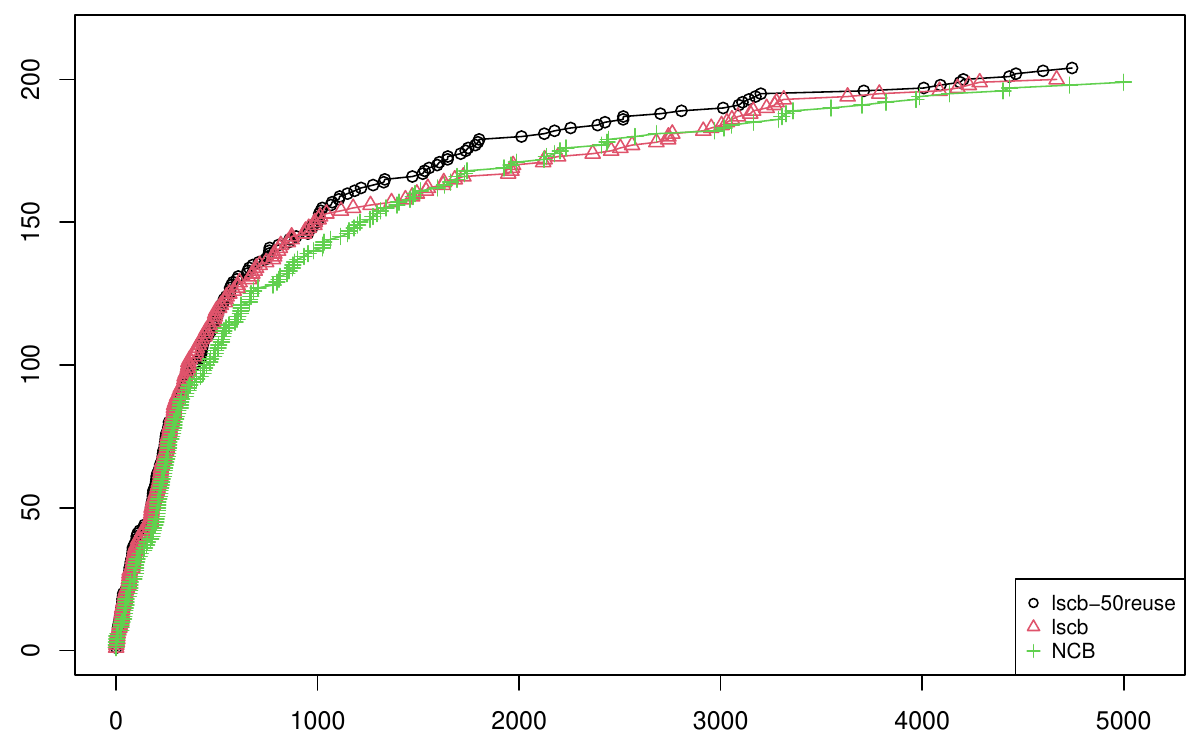}
    \caption{\label{fig:cdf-glucose-2018}\glucose variants in the SAT Competition 2018}
    \end{subfigure}

    \begin{subfigure}{1\textwidth}
    \includegraphics[width=\textwidth]{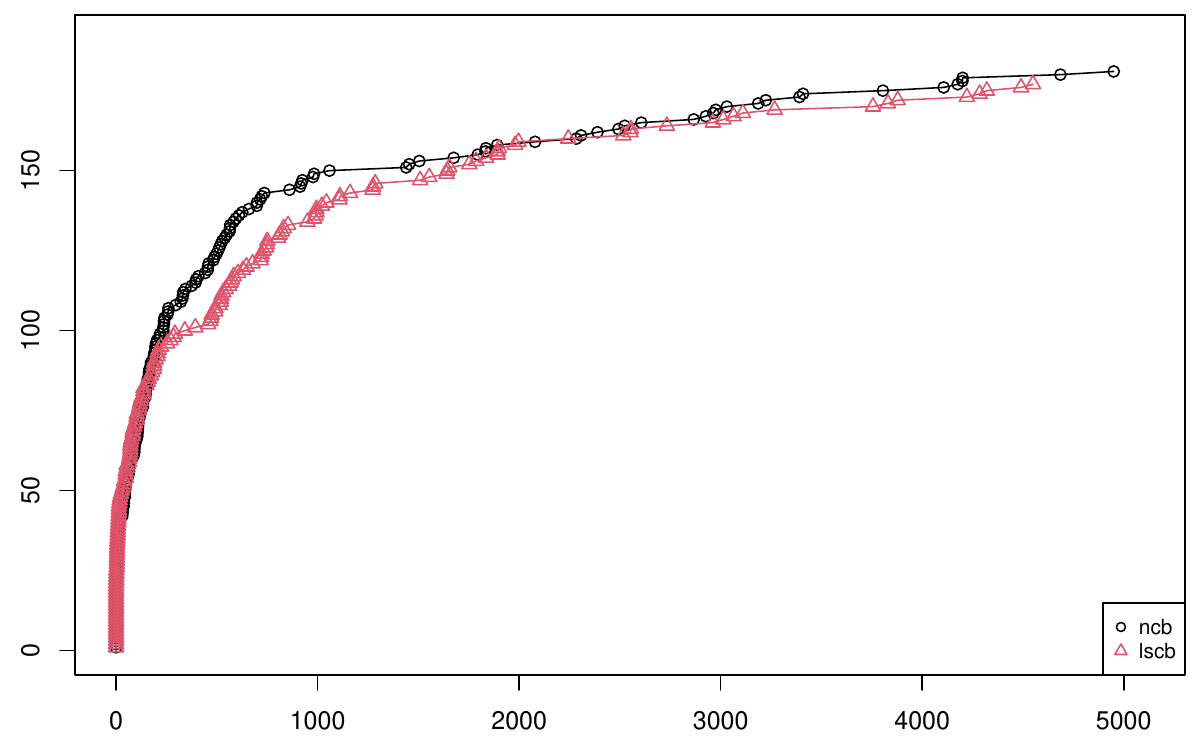}
    \caption{\label{fig:cdf-glucose-2023}\glucose variants in the SAT Competition 2023}
    \end{subfigure}

    \caption{\label{fig:cdf-glucose}CDF of the different \glucose (without strategy adapting) versions. The constant indicates when chronological backtracking is triggered instead of backjumping: We apply chronological backtracking when NCB would require jumping back more than 100 levels by default, like in \cadical. In the SAT Competition 2018, the version that triggers chronological backtracking for more than 50 levels performs best.}
\end{figure}
\newpage
\section{Conclusion}
We introduce a lazy reimplication procedure to be used in CDCL with (variants of) chronological backtracking.
We particularly focus on the definitions of weak chronological backtracking (WCB), restoring strong chronological backtracking (RSCB), eager strong chronological backtracking (ESCB), and lazy strong chronological backtracking (LSCB).
Our invariant properties on these backtracking variants exploit watched literals. We prove that our approach of
lazy reimplication in strong chronological backtracking (LSCB) yields a sound and complete SAT solving method. Our implementation in \napsat, and its integration with \cadical and \glucose, gives practical evidence that
LSCB is significantly easier to implement than ESCB, while also propagating fewer literals than RSCB, and providing better guarantees than WCB.

In the future, we intend to extend our LSCB method to reason over virtual literal levels, that is, levels of missed lower implications if such a clause is detected. We believe such an extension would allow to converge closer to the guarantees of ESCB, while mitigating both algorithmic complexities and reimplication costs. Further, we will explore the integration of chronological backtracking variants in the context of SMT, as a robust approach to handling arbitrary incremental clauses and expensive theory propagations.

\bibliography{bibliography}

\end{document}

%% file: sat-macros.tex
\definecolor{Green}{rgb}{0.13, 0.80, 0.13}

\newcommand{\blue}[1] {
	\protect\leavevmode
  \begingroup
	\color{blue}%
    #1%
	\endgroup
}
\newcommand{\red}[1] {
	\protect\leavevmode
  \begingroup
	\color{red}%
    #1%
	\endgroup
}
\newcommand{\green}[1] {
	\protect\leavevmode
  \begingroup
	\color{Green}%
    #1%
	\endgroup
}

\newcommand*{\trail}{\tau}
\newcommand*{\q}{\omega}
\newcommand*{\p}{\pi}

\newcommand*{\wl}{\mathrm{WL}} 
\newcommand*{\reason}{\rho} 
\newcommand*{\level}{\delta}
\newcommand*{\lazy}{\lambda} 
\newcommand*{\onLevel}[1]{@#1} 
\newcommand*{\napsat}[0]{\textsc{NapSAT}\xspace} 
\newcommand*{\cadical}[0]{\textsc{CaDiCaL}\xspace} %
\newcommand*{\glucose}[0]{\textsc{Glucose}\xspace} %

%% file: Figures/trail-6.tex
\draw[thick] (0, 0) -- node[below, yshift = -0.1cm] {\rotatebox{270}{decision}} (1, 0);
\draw[thick] (0, 0.000000) -- (0, 0.750000);
\draw[thick] (0, 0.750000) -- node[below] {$\level = 1$} node[above] {$v_{1}$} (1, 0.750000);
\draw[thick] (1, 0) -- node[below, yshift = -0.1cm] {\rotatebox{270}{decision}} (2, 0);
\draw[thick] (1, 0.750000) -- (1, 1.500000);
\draw[thick] (1, 1.500000) -- node[below] {$\level = 2$} node[above] {$v_{2}$} (2, 1.500000);
\draw[thick] (2, 0) -- node[below, yshift = -0.1cm] {$C_7$} (3, 0);
\draw[thick] (2, 1.500000) -- (2, 0.750000);
\draw[thick] (2, 0.750000) -- node[above] {$\neg v_{3}$} (3, 0.750000);
\draw[thick] (3, 0) -- node[below, yshift = -0.1cm] {$C_3$} (4, 0);
\draw[thick] (3, 0.750000) -- node[above] {$v_{5}$} (4, 0.750000);
\draw[thick] (4, 0) -- node[below, yshift = -0.1cm] {$C_6$} (5, 0);
\draw[thick] (4, 0.750000) -- (4, 1.500000);
\draw[thick] (4, 1.500000) -- node[above] {$\neg v_{6}$} (5, 1.500000);
\draw[thick, red] (5, 0) -- node[below, yshift = -0.1cm] {$\red{C_5}$} (6, 0);
\draw[thick, red] (5, 1.500000) -- node[above, red] {$\bot$} (6, 1.500000);

\draw[thick, dotted] (3, -1.500000) -- (3, 2.50000) node[right, yshift=-0.2cm] {$\q \rightarrow$} node[left, yshift=-0.2cm] {$\leftarrow \trail$};

\foreach \x in {0,1,...,6}
  \draw[thick] (\x,3pt)--(\x,-3pt);

%% file: Figures/trail-7-wcb.tex
\draw[thick] (0, 0) -- node[below, yshift = -0.1cm] {\rotatebox{270}{decision}} (1, 0);
\draw[thick] (0, 0.000000) -- (0, 0.750000);
\draw[thick] (0, 0.750000) -- node[below] {$\delta = 1$} node[above] {$v_{1}$} (1, 0.750000);
\draw[thick] (1, 0) -- node[below, yshift = -0.1cm] {$C_7$} (2, 0);
\draw[thick] (1, 0.750000) -- node[above] {$\neg v_{3}$} (2, 0.750000);
\draw[thick] (2, 0) -- node[below, yshift = -0.1cm] {$C_3$} (3, 0);
\draw[thick] (2, 0.750000) -- node[above] {$v_{5}$} (3, 0.750000);

\draw[thick, dotted] (2, -1.500000) -- (2, 1.75000) node[right, yshift=-0.2cm] {$\q \rightarrow$} node[left, yshift=-0.2cm] {$\leftarrow \trail$};

\foreach \x in {0,1,...,3}
  \draw[thick] (\x,3pt)--(\x,-3pt);

%% file: Figures/trail-7-rscb.tex
\draw[thick] (0, 0) -- node[below, yshift = -0.1cm] {\rotatebox{270}{decision}} (1, 0);
\draw[thick] (0, 0.000000) -- (0, 0.750000);
\draw[thick] (0, 0.750000) -- node[below] {$\delta = 1$} node[above] {$v_{1}$} (1, 0.750000);
\draw[thick] (1, 0) -- node[below, yshift = -0.1cm] {$C_7$} (2, 0);
\draw[thick] (1, 0.750000) -- node[above] {$\neg v_{3}$} (2, 0.750000);
\draw[thick] (2, 0) -- node[below, yshift = -0.1cm] {$C_3$} (3, 0);
\draw[thick] (2, 0.750000) -- node[above] {$v_{5}$} (3, 0.750000);

\draw[thick, dotted] (1, -1.500000) -- (1, 1.75000) node[right, yshift=-0.2cm] {$\q \rightarrow$} node[left, yshift=-0.2cm] {$\leftarrow \trail$};

\foreach \x in {0,1,...,3}
  \draw[thick] (\x,3pt)--(\x,-3pt);

%% file: Figures/trail-6-escb.tex
\draw[thick] (0, 0) -- node[below, yshift = -0.1cm] {\rotatebox{270}{decision}} (1, 0);
\draw[thick] (0, 0.000000) -- (0, 0.750000);
\draw[thick] (0, 0.750000) -- node[below] {$\level = 1$} node[above] {$v_{1}$} (1, 0.750000);
\draw[thick] (1, 0) -- node[below, yshift = -0.1cm] {$C_4$} (2, 0);
\draw[thick] (1, 0.75) -- node[above] {$v_{2}$} (2, 0.75);
\draw[thick] (2, 0) -- node[below, yshift = -0.1cm] {$C_7$} (3, 0);
\draw[thick] (2, 0.750000) -- node[above] {$\neg v_{3}$} (3, 0.750000);
\draw[thick] (3, 0) -- node[below, yshift = -0.1cm] {$C_3$} (4, 0);
\draw[thick] (3, 0.750000) -- node[above] {$v_{5}$} (4, 0.750000);
\draw[thick] (4, 0) -- node[below, yshift = -0.1cm] {$C_6$} (5, 0);
\draw[thick] (4, 0.75) -- node[above] {$\neg v_{6}$} (5, 0.75);
\draw[thick, red] (5, 0) -- node[below, yshift = -0.1cm] {$\red{C_5}$} (6, 0);
\draw[thick, red] (5, 0.75) -- node[above, red] {$\bot$} (6, 0.75);

\draw[thick, dotted] (3, -1.500000) -- (3, 1.750000) node[right, yshift=-0.2cm] {$\q \rightarrow$} node[left, yshift=-0.2cm] {$\leftarrow \trail$};

\foreach \x in {0,1,...,6}
  \draw[thick] (\x,3pt)--(\x,-3pt);

%% file: Algorithms/cdcl.tex
\begin{algorithmic}[1]
    \State $\p = \trail = \q = \p^d = \emptyset$
    \State $\forall \ell.\ \delta(\ell) = \infty$
    \State $\forall \ell.\ \wl(\ell) = \emptyset$
    \State $\forall \ell.\ \reason(\ell) = \lazy(\ell) = \square$
    \Procedure{CDCL}{$F$}
        \For{$C \in F$}
        \Comment{\parbox{0.45\textwidth}{Fill the watcher lists}}
            \State $c_1, c_2 \gets $ two literals in $C$
            \State $\wl \gets \wl[c_1 \gets \wl(c_1) \cup \{C\}][c_2 \gets \wl(c_2) \cup \{C\}]$
        \EndFor
        \While {$\top$}
            \State $C \gets \Call{BCP}{ }$ \Comment{\parbox{0.45\textwidth}{\Cref{alg:bcp}}}
            \If{$C = \top$}
                \If{$|\p| = |\mathcal{V}|$}
                \Comment{\parbox{0.45\textwidth}{All variables are assigned}}
                    \State \Return \textsc{SAT}
                \EndIf
                \State $\ell \gets \Call{Decide}{ }$
                \State $\q \gets \q \cdot \ell, \p^d \gets \p^d \cdot \ell, \level \gets \level[\ell \gets |\p^d|]$
                \State \Continue
            \EndIf
            \State $D \gets \Call{Analyze}{C}$
            \Comment{\parbox{0.45\textwidth}{\Cref{alg:conflict-analysis}}}
            \If{$\level(D) = 0$}
                \State \Return \textsc{UNSAT}
            \EndIf
            \State $d \gets $ any level between $\level(D) - 1$ and the second highest level of $D$
            \State $\Call{Backtrack}{d}$ \Comment{\parbox{0.45\textwidth}{\Cref{alg:backtracking}}}
            \State $\ell \gets$ the unassigned literals in $D$
            \State $c_2 \gets $ the second highest literal in $D$
            \State $\q \gets \q \cdot \ell, \level \gets \level[\ell \gets \level(C \setminus \{\ell\})], \reason \gets \reason[\ell \gets D]$
            \State $F \gets F \cup \{D\}$
            \Comment{\parbox{0.45\textwidth}{Does nothing if $C = D$}}
            \State $\wl \gets \wl[\ell \gets \wl(\ell) \cup \{D\}][c_2 \gets \wl(c_2) \cup \{D\}]$
        \EndWhile
    \EndProcedure
\end{algorithmic}

%% file: Algorithms/propagate-lit.tex
\begin{algorithmic}[1]
  \Procedure{PropagateLiteral}{$\ell$}
    \State $c_1 \gets \neg \ell$
    \For{$C \in \wl[c_1]$}
      \State $c_2 \gets $ the other watched literal in $C$
      \If {$c_2 \in \p \land $ \textcolor{blue}{$\left[\level(c_2) \leq \level(c_1) \lor \level(\lambda(c_2) \setminus \{c_2\}) \leq \level(c_1)\right]$}}
        \State \Continue
        \label{line:prop:watched-literal-satisfied}
      \EndIf

      \State $r \gets$ \Call{SearchReplacement}{$C, c_1, c_2$}
      \State $\wl \gets \wl[c_1 \gets \wl(c_1) \setminus \{C\}][r \gets \wl(r) \cup \{C\}]$
      \If {$\neg r \notin \p$}
        \State \Continue
        \label{line:prop:watched-literal-replaced}
      \EndIf
      \label{line:prop:watched-literal-replaced-higher-level}
      \If {$\neg c_2 \in \p$}
      \Comment{\parbox{0.45\textwidth}{Conflict}}
        \State \Return $C$
        \label{line:prop:conflict}
      \EndIf
      \textcolor{blue}{
      \If {$c_2 \in \p$}
        \If {$\level(c_2) > \level(r) \land \level(\lazy(c_2) \setminus \{c_2\}) > \level(r)$}
        \label{line:prop:better-lower-implication-test}
          \State $\lazy \gets \lazy[c_2 \gets C]$
          \Comment{\parbox{0.45\textwidth}{New or improved MLI}}
        \EndIf
        \State \Continue
        \label{line:prop:missed-implication}
      \EndIf
      }
      \State $\q \gets \q \cdot c_2,\reason \gets \reason[c_2\gets C], \level \gets \level[c_2 \gets \level(r)]$
      \label{line:prop:propagate}
    \EndFor
    \State \Return $\top$
  \EndProcedure
\end{algorithmic}

%% file: Algorithms/bcp.tex
\begin{algorithmic}[1]
  \Procedure{BCP}{}
    \While {$\q \neq \emptyset$}
      \State $\ell \gets \Call{First}{\q}$
      \State $C \gets \Call{PropagateLiteral}{\ell}$
      \If {$C \neq \top$}
        \State \Return $C$
      \EndIf
      \State $\q \gets \q \setminus \{\ell\}, \trail \gets \trail \cdot \ell$
    \EndWhile
  \State \Return $\top$
  \EndProcedure

\end{algorithmic}

%% file: Algorithms/backtracking.tex
\begin{algorithmic}[1]
  \Procedure{Backtrack}{$d$}
    \State $\Lambda \gets \emptyset$
    \Comment{\parbox{0.45\textwidth}{$\Lambda$ is the set that will be reimplied}}
    \State $\p = \trail \cdot \q$
    \For{$\ell \in \p$}
        \If{$\level(\ell) > d$}
            \textcolor{blue}{
            \If{$\level(\lazy(\ell) \setminus \{\ell\}) \leq d$}
                \State $\Lambda \gets \Lambda \cup \{\lazy(\ell)\}$
                \Comment{\parbox{0.45\textwidth}{Store the MLI for later}}
            \EndIf
            }
            \State $\p \gets \p \setminus \{\ell\}$
            \Comment{\parbox{0.45\textwidth}{Unassign $\ell$}}
            \State $\level \gets \level[\ell \gets \infty], \reason \gets \reason[\ell \gets \square]$
            \textcolor{blue}{
            \State $\lazy \gets \lazy[\ell \gets \square]$
            \Comment{\parbox{0.45\textwidth}{$\lazy(\ell)$ is either used, or no longer valid}}
            }
        \EndIf
    \EndFor
    \State $\p^d \gets \p \cap \p^d$
    \Comment{\parbox{0.45\textwidth}{Remove the unassigned literals}}
    \State $\trail \gets \p \cap \trail$
    \State $\q \gets \p \setminus \trail$
    \textcolor{blue}{\For {$C \in \Lambda$}
    \Comment{\parbox{0.45\textwidth}{Reimplying the MLI}}
      \State $\ell \gets $ the unassigned literal in $C$
      \State $\q \gets \q \cdot \ell, \reason \gets \reason[\ell \gets C], \level \gets \level[\ell \gets \level(C \setminus \{\ell\})]$
    \EndFor}
  \EndProcedure
\end{algorithmic}

%% file: Algorithms/conflict-analysis.tex
\begin{algorithmic}[1]
  \Procedure{Analyze}{$C$}
    \State $\pi \gets \trail \cdot \q$
    \Comment{\parbox{0.45\textwidth}{Array version of the trail.}}
    \State $D \gets C$
    \Comment{\parbox{0.45\textwidth}{Current learned clause.}}
    \State $n \gets \left|\left\{\ell: \ell \in D \land \level(\ell) = \level(D)\right\}\right|$
    \Comment{\parbox{0.45\textwidth}{Number of literals at the highest level.}}
    \While {$\top$}
      \State $\ell \gets $ the last literal in $\p$ falsified in $D$ at level $\level(D)$
      \If {$n = 1 \blue{\land \lazy(\ell) = \square}$}
        \State \Return $D$
      \EndIf

      \State $C' \gets \reason(\ell)$
      \If {\textcolor{blue}{$\lazy(\ell) \neq \square$}}
        \State \textcolor{blue}{$C' \gets \lazy(\ell)$}
      \EndIf
      \State $D \gets D\otimes_\ell C'$
      \State $n \gets \left|\left\{\ell: \ell \in D \land \level(\ell) = \level(D)\right\}\right|$
    \EndWhile
  \EndProcedure
\end{algorithmic}

%% file: Figures/trail-analysis.tex
\draw[thick] (0, 0) -- node[below, yshift = -0.1cm] {\rotatebox{270}{decision}} (1, 0);
\draw[thick] (0, 0.000000) -- (0, 0.750000);
\draw[thick] (0, 0.750000) -- node[below] {$\delta = 1$} node[above] {$\neg v_{1}$} (1, 0.750000);
\draw[thick] (1, 0) -- node[below, yshift = -0.1cm] {$C_1$} (2, 0);
\draw[thick] (1, 0.750000) -- node[above] {$\neg v_{2}$} (2, 0.750000);
\draw[thick] (2, 0) -- node[below, yshift = -0.1cm] {\rotatebox{270}{decision}} (3, 0);
\draw[thick] (2, 0.750000) -- (2, 1.500000);
\draw[thick] (2, 1.500000) -- node[below] {$\delta = 2$} node[above] {$\neg v_{3}$} (3, 1.500000);
\draw[thick] (3, 0) -- node[below, yshift = -0.1cm] {$C_7$} (4, 0);
\draw[thick] (3, 1.500000) -- (3, 0.750000);
\draw[thick] (3, 0.750000) -- node[above] {$v_{4}$} (4, 0.750000);
\draw[thick] (4, 0) -- node[below, yshift = -0.1cm] {$C_2$} (5, 0);
\draw[thick] (4, 0.750000) -- (4, 1.500000);
\draw[thick] (4, 1.500000) -- node[above] {$\neg v_{5}$} (5, 1.500000);
\draw[thick] (5, 0) -- node[below, yshift = -0.1cm] {$C_3$} (6, 0);
\draw[thick] (5, 1.500000) -- (5, 0.750000);
\draw[thick] (5, 0.750000) -- node[above] {$\neg v_{6}$} (6, 0.750000);
\draw[thick] (6, 0.750000) -- (6, 1.500000);
\draw[thick] (6, 0) -- node[below, yshift = -0.1cm] {$C_4$} (7, 0);
\draw[thick] (6, 1.500000) -- node[above] {$v_{7}$} (7, 1.500000);
\draw[thick, red] (7, 0) -- node[below, yshift = -0.1cm] {$C_5$} (8, 0);
\draw[thick, red] (7, 1.500000) -- node[above, red] {$\bot$} (8, 1.500000);

\draw[thick, dotted] (4, -1.500000) -- (4, 2.500000) node[right, yshift=-0.2cm] {$\q \rightarrow$} node[left, yshift=-0.2cm] {$\leftarrow \trail$};

\foreach \x in {0,1,...,8}
  \draw[thick] (\x,3pt)--(\x,-3pt);